\documentclass[runningheads]{llncs}
\usepackage{amsmath}
\usepackage{mathtools}
\usepackage{mathrsfs}
\usepackage{amsfonts}
\usepackage{amssymb}
\usepackage{multirow}
\usepackage{multicol}
\usepackage[normalem]{ulem}
\usepackage{siunitx}
\usepackage[utf8]{inputenc}
\usepackage{todonotes}
\usepackage{booktabs}
\usepackage{xcolor}
\usepackage{nameref}
\usepackage{caption}
\usepackage{subcaption}
\usepackage{wrapfig}

\usepackage{tikz}
\usetikzlibrary{automata, positioning, arrows, shapes.geometric, shadows,positioning,angles,quotes,patterns, decorations.text, calc}
\usepackage[ruled,vlined,linesnumbered]{algorithm2e}
\SetKwInOut{Require}{Require} 
\SetKw{Input}{Input:}\SetKw{Output}{Output:}
\SetKwRepeat{DoWhile}{do}{while}

\SetCommentSty{mycommfont}

\newcommand{\ipath}{\pi}
\newcommand{\schr}{\mu}

\newcommand{\dist}{\nu}

\newcommand{\vect}[1]{\boldsymbol{#1}}
\newcommand{\dfa}{\mathcal{A}}
\newcommand{\mdp}{\mathcal{M}}
\newcommand{\tlab}{\mathsf{ct}}
\newcommand{\tmdp}{\mdp^\tlab }

\DeclareMathOperator*{\argmax}{arg\,max}
\DeclareMathOperator*{\argmin}{arg\,min}
\DeclareMathOperator*{\utl}{\mathsf{U}} 
\DeclareMathOperator*{\nxt}{\mathsf{X}} 
\DeclareMathOperator*{\evt}{\mathsf{F}} 
\DeclareMathOperator*{\alw}{\mathsf{G}} 

\usepackage{marginnote}

\usepackage{ulem}

\newcommand{\parahead}[1]{\medskip \noindent\textbf{#1.}}
\newcommand*{\TickSize}{2pt}%

\makeatletter
\providecommand*{\cupdot}{%
  \mathbin{%
    \mathpalette\@cupdot{}%
  }%
}
\newcommand*{\@cupdot}[2]{%
  \ooalign{%
    $\m@th#1\cup$\cr
    \sbox0{$#1\cup$}%
    \dimen@=\ht0 %
    \sbox0{$\m@th#1\cdot$}%
    \advance\dimen@ by -\ht0 %
    \dimen@=.5\dimen@
    \hidewidth\raise\dimen@\box0\hidewidth
  }%
}

\providecommand*{\bigcupdot}{%
  \mathop{%
    \vphantom{\bigcup}%
    \mathpalette\@bigcupdot{}%
  }%
}
\newcommand*{\@bigcupdot}[2]{%
  \ooalign{%
    $\m@th#1\bigcup$\cr
    \sbox0{$#1\bigcup$}%
    \dimen@=\ht0 %
    \advance\dimen@ by -\dp0 %
    \sbox0{\scalebox{1.5}{$\m@th#1\cdot$}}%
    \advance\dimen@ by -\ht0 %
    \dimen@=.5\dimen@
    \hidewidth\raise\dimen@\box0\hidewidth
  }%
}
\makeatother
\begin{document}
\title{Multi-Objective Task Assignment and Multiagent Planning with Hybrid GPU-CPU Acceleration\thanks{This work was partially supported by Australian Defence Science and Technology Group under the Artificial Intelligence for Decision Making 2022 Initiative scheme.}}
%
%
\author{Thomas Robinson\orcidID{0000-0002-6150-2587} \and
Guoxin Su\orcidID{0000-0002-2087-4894}} 
%
\authorrunning{T. Robinson, G. Su}
\titlerunning{Multi-Objective Task Assignment and Multiagent Planning}
%
%
\institute{University of Wollongong, NSW 2522, Australia \\
\email{tmr463@uowmail.edu.au}\quad
\email{guoxin@uow.edu.au}}
%
\maketitle  
%
\begin{abstract}
Allocation and planning with a collection of tasks and a group of agents is an important problem in multiagent systems.
One commonly faced bottleneck is scalability, as in general the multiagent model increases exponentially in size with the number of agents.
We consider the combination of random task assignment and multiagent planning under multiple-objective constraints, and show that this problem can be decentralised to individual agent-task models.
We present an algorithm of point-oriented Pareto computation, which checks whether a point corresponding to given cost and probability thresholds for our formal problem is feasible or not. If the given point is infeasible, our algorithm finds a Pareto-optimal point which is closest to the given point. We provide the first multi-objective model checking framework that simultaneously uses GPU and multi-core acceleration. Our framework manages CPU and GPU devices as a load balancing problem for parallel computation. Our experiments demonstrate that parallelisation achieves significant run time speed-up over sequential computation.

\keywords{Multiagent System \and Task Assignment \and Planning \and Probabilistic Model Checking \and GPU and Multi-Core Acceleration}
\end{abstract}
\section{Introduction}
Markov Decision Processes (MDPs) \cite{puterman2014markov} are a fundamental model for multiagent planning in stochastic environments, where actions of an agent at a state may lead to uncertain outcomes.
Multiagent task allocation and planning is concerned with enabling a group of agents to divide up tasks amongst themselves and carry out their planning and execution. 
Scalability is a commonly faced bottleneck for this kind of problems, as in general an MDP that models a multiagent system (MAS) increases exponentially in size with a linear increment in the number of agents in the system \cite{boutilier1996planning}. 


Probabilistic model checking (PMC) is a verification technique to establish rigorous guarantees about the correctness of real-life stochastic systems \cite{baier201910}. 
PMC provides methods to compute the optimal values of reachability rewards for an MDP, and the optimal probabilities that an MDP satisfies properties formalised with Linear Temporal Logic (LTL).
A fragment of LTL, called co-safe LTL, has a deterministic finite-state automaton (DFA) representation \cite{KupfermanOrna2001MCoS}, and thus is suitable to specify tasks that must be completed in finite time.
Task execution in finite time is important in multiagent planning because we typically want to re-use the agents to execute further tasks.  
 
In practice, coordination of agents usually involves conflicting solutions to the multiple objectives that an MAS is required to satisfy, for example, agents may need to balance execution time with energy consumption. 
When simultaneous verification of multiple objectives is concerned, we require the multi-objective MDP (MOMDP) \cite{roijers2013survey}
whose reward structure specifies reward vectors (rather than scalars). 
The solution space of an MOMDP is a convex polytope \cite{forejt2012pareto,etessami2007multi}, which makes the MOMDP model checking problem tractable. 
Currently three kinds of queries are considered in MOMDP model checking \cite{forejt_automatic_2011}: 
The achievability query is the most basic query, which asks whether there exists a scheduler to meet all objective thresholds; the numerical query is a numerical variant of the first query, which computes the optimal value of one objective while meeting all other objective thresholds; the Pareto query is the most expensive query, which computes approximately the Pareto curve of all objectives.

The classical assignment problem finds an assignment, namely a one-to-one mapping from tasks to agents, which results in a maximal assignment reward.
The multi-objective assignment problem is to determine an assignment such that the vectorised assignment reward is Pareto optimal.
The classical assignment problem can be solved efficiently (e.g., using the Hungarian algorithm \cite{kuhn1955hungarian}), but the multi-objective assignment problem is much harder \cite{Ulungu1994}.
The multi-objective random assignment (MORA) problem pursues a randomised distribution over assignments such that the expected assignment reward is Pareto optimal.

The combination of (single-objective) task assignment and multiagent planning has been considered for non-stochastic agent models (i.e., transition systems) \cite{schillinger_simultaneous_2018} and stochastic agent models 
(i.e., MDPs) \cite{faruq_simultaneous_2018}.
In this paper, we extend MOMDP model checking to a setting of multi-objective random assignment and planning (MORAP) in an MAS, and present a novel implementation with hybrid GPU-CPU acceleration. Our main contributions are as follows:
\begin{itemize}
\vspace{-1mm}
    \item 
    We show the convexity of our formal problem (MORAP), and that a practical approach to solve this problem can rely on a decentralised model, which avoids the exponential model size growth with agent-task numbers.
    \item Our main algorithm is a new point-oriented Pareto computation complementing the existing achievability and Pareto queries \cite{forejt2012pareto}. 
    For a given point corresponding to cost and probability thresholds, our algorithm finds a point which is feasible for the MORAP problem and closest to the given point under a general vector norm.
    \item 
    To the best of our knowledge, we provide the first multi-objective model checking framework that utilises simultaneous GPU and multi-core acceleration.
    Our framework manages CPU and GPU devices as a load balancing problem for parallel computation. We evaluate the performance of our implementation in a smart-warehouse example.
\end{itemize}

The remainder of this paper is organised as follows: Section \ref{sec:preliminaries} provides the preliminaries for the problem; Section \ref{sec:approach} gives the approach to the problem, model construction and algorithms; Section \ref{sec:implementation} provides details on the hybrid implementation and parallel architecture; Section \ref{sec:experiments} analyses the performance of our approach; Section \ref{sec:related} provides related work; and finally Section \ref{sec:conclusion} concludes the paper.
Formal proofs of theorems are included in the appendix of the long version of this paper \cite{long_version}. 

\section{Preliminaries}\label{sec:preliminaries}
\parahead{Deterministic Finite Automata}
A \emph{deterministic finite automaton} (DFA) $\dfa $ is given by the tuple $(Q, q_0, Q_F, \Sigma, \delta)$ where (i) $Q$ is a set of locations,
(ii) $q_0\in Q$ is an initial location, (iii) $Q_F\subseteq Q$ is a set of accepting locations, (iv) $\Sigma = 2^{AP}$ (where ${AP}$ is a non-empty set of atomic propositions) is the alphabet, and (v) $\delta: S\times \Sigma \to S$ is the transition function.
If $\delta(q,W)= q'$ for some $W\subseteq AP$, we call $q$ a \emph{predecessor} of $q'$ and $q'$ a \emph{successor} of $q$. Let $\mathrm{pre}(q)$ and $\mathrm{suc}(q)$ denote the set of predecessors or successors of $q$, respectively.
A location $q$ is a \emph{sink} if $\mathrm{suc}(q)=\{q\}$. 
In this paper, it suffices to consider DFAs {whose accepting locations are sinks}.
A run in $\mathcal{A}$ is a sequence of locations $q_1,\ldots, q_m$ such that $q_{i+1}\in \mathrm{suc}(q_i)$ for all $1\leq i\leq m-1$. 
We call $q$ a \emph{trap} if there is no run to any $q'\in Q_F$ from it. Let $Q_R$ be the set of traps in $\mathcal{A}$.
%

\parahead{Co-Safe LTL}
LTL is a compact representation of linear time properties. The syntax of LTL is $\varphi ::= \top \mid \mathtt{a} \mid \neg \varphi \mid \varphi \land \varphi \mid \nxt\varphi \mid \varphi \utl \varphi$, 
where $\mathtt{a} \in AP$. The operators $\nxt$ and $\utl$ stand for ``next" and ``until'', respectively. Let  $\evt\varphi: = \top \utl \varphi $, and $ \alw\varphi:=\neg \evt \neg \varphi$.
The semantic relationship $\sigma\models \varphi$ where $\sigma\in \Sigma^\omega $ is standard
where $\Sigma^\omega$ denotes the set of all infinite words over $\Sigma$.  
We are interested in the \emph{co-safe} fragment of LTL formulas.
Informally, $\varphi$ is {co-safe} if any $\sigma$  such that $\sigma\models \varphi$ includes some \emph{good prefix} (which is accepting in some DFA) 
$pref_{good}(\varphi)$ denoted ${acc}(\dfa)$.
Syntactically, any LTL formula containing only the temporal operators $\nxt$ (\emph{next}), $\utl$ (\emph{until}), and $\evt$ (\emph{eventually}) in \emph{positive normal form} (PNF) is co-safe.
A formal characterisation in the semantic level is included the appendix of \cite{long_version}.

\parahead{Markov Decision Process}
    A (labelled) MDP is given by the tuple 
    $\mathcal{M}= (S, s_0, A, {P},L)$ where (i) $S$ is a finite nonempty state space, (ii) $s_0\in S$ is an initial state, (iii) $A$ is a set of actions, (iv) ${P} : S\times A\times S \to [0,1]$ is a transition probability function such that $\sum_{s'\in S} {P}(s,a,s') \in \{0,1\}$, and (v) $L: S\to \Sigma$ is a labelling function. 
    Let $A(s) = \{a \in A \mid \sum_{s'\in S} {P}(s,a,s')=1\}$, i.e., $A(s)$ is the set of \emph{enabled} actions at $s$.  
    The size of $\mdp$ is $|\mdp|= |S|+|P|$, where $|P|=|\{(s,a,s')\in S\times A\times S\mid P(s,q,s')>0\}|$.
    A \emph{reward function or structure} for $\mathcal{M}$ is a function $\rho: \{(s,a) \in S\times A \mid a\in A(s)\}\to \mathbb{R}$.
We write $\mathcal{M}[\rho]$ to explicitly indicate the reward structure $\rho$ for $\mathcal{M}$.
%
An (infinite) \emph{path} $\ipath$ is a sequence $s_1a_1s_2a_2\ldots$ such that ${P}(s_{i},a_i,s_{i+1})>0$ for all $i\geq 1$.
Let $L(\pi)$ denote the word $L(s_1)L(s_2)\ldots \in \Sigma^\omega$.
Let $\mathrm{IPath}$ be the set of paths in $\mathcal{M}$ and $\mathrm{IPath}(s)$ be the subset of $\mathrm{IPath}$ containing the paths originating from $s$.
The set of probability distributions over $A$ is denoted by $Dist(A)$.
A memoryless \emph{scheduler} (or scheduler for short) 
for $\mathcal{M}$, is a mapping $\schr: s\mapsto Dist(A(s))$ for all $s\in S$. 
If $\schr$ is a \emph{simple} (or pure) if $\schr(s)(a)=1$ for each $s\in S$ and some $a\in A(s)$.
The set of schedulers (resp., simple schedulers) is denoted by $Sch(\mathcal{M})$ (resp., $Sch_\mathrm{S}(\mathcal{M})$).
%

\parahead{Reachability Reward}
Given any LTL formula $\evt\!B$ with $B$ being a Boolean formula, let ${\rho}(\ipath|\evt\!B)= \sum_{i=1}^n \rho(s_i,a_i)$ where $\ipath = s_1a_1s_2a_2\ldots\in \mathrm{IPath}(s_1)$ and $n$ is the smallest number such that $L(s_n)\models B$ and $L(s_i)\not\models B$ for all $i <n$; if such $n$ does not exist, let $\rho(\ipath|\evt\!B)=\infty$. 
Let $\mathbf{Pr}^{\mathcal{M},\schr}$ be the probability measure over paths in $\mathrm{IPath}(s)$.\footnote{This probability measure is defined on the discrete-time Markov chain induced by the scheduler $\schr$ of $\mdp$ (c.f.\ Definition 10.92 in \cite{baier_principles_2008}).}
The expectation $\mathbf{E}^{\mathcal{M}[\rho],\schr}(\evt B) \doteq \int_{\ipath} \rho(\ipath|\evt\!B) \mathrm{d}\mathbf{Pr}^{\mathcal{M},\schr}$, a.k.a.\ \emph{reachability reward} \cite{kwiatkowska2022probabilistic}, is the {expected reward accumulated in a path of $\mathcal{M}$ under $\schr$ until reaching states satisfying $B$.} 
We say $\mdp[\rho]$ is \emph{reward-finite} w.r.t.\ $\evt\!B$ if $ \sup_{\schr\in {Sch}(\mdp)} \mathbf{E}^{\mathcal{M}[\rho],\schr}(\evt B)<\infty$. 

\parahead{Product MDP}
Given $\mathcal{M}= (S, s_0, A, {P}, L)$ and $\dfa = (Q,q_0, Q_F, \Sigma, \delta)$, a \emph{product MDP} is a tuple $\mathcal{M} \otimes \dfa =(S\times Q, (s_0, q_0), A, {P}', L')$
    where (i) $P':S\times Q\times {A} \times S\times Q \to [0,1]$ is a transition probability function such that
	\begin{equation*}
    	{P}'( s, q, a, s', q' ) = \left\{ 
    	\begin{array}{ll}
    	{P}(s, a, s') & \quad \text{if } q' = \delta(q, L(s'))\\
    	0 & \quad \text{otherwise}
    	\end{array}\right.
	\end{equation*}
    and (ii) $L':S\times Q\to 2^\Sigma$  
    is a labelling function s.t.\ $L'(s,q) = L(s)$.
Let $\mdp[\rho]\otimes \dfa$ refer to $(\mdp\otimes \dfa)[\rho]$ where $\rho(s,q,a) = \rho(s,a)$ for all $(s,q)\in S\times Q, a\in A(s)$.

%
\parahead{Geometry}
For a point (i.e., vector) $\vect{v}
\in \mathbb{R}^n$  for some $n$, let $v_i$ denote the $i^\mathrm{th}$ element of $\vect{v}$.
A \emph{weight vector} $\vect{w}$ is a vector such that $w_i\geq 0$ and $\sum_{i=1}^n w_i=1$.
The \emph{dot product} of $\vect{v}$ and $\vect{u}$, denoted $\vect{v}\cdot \vect{u}$, is the sum $\sum_{i=1}^n v_iu_i$.
For a set $\Phi=\{\vect{v}_1, \ldots, \vect{v}_m\} \subseteq \mathbb{R}^n$,
a \emph{convex combination} in $\Phi$ is $\sum_{i=1}^m w_i \vect{v}_i$ for some weight vector $\vect{w}\in \mathbb{R}^m$.
The \emph{downward closure} of the \emph{convex hull} of $\Phi$, denoted ${down}(\Phi)$, is the set of vectors such that for any $\vect{u}\in {down}(\Phi)$ there is a convex combination $\vect{v} = w_1\vect{v}_1+\ldots +w_m\vect{v}_m$ such that 
$u_i \leq v_i$.
Let $\Psi\subseteq \mathbb{R}^n$ be any downward closure of points. 
We say $\vect{u}$ \emph{dominates $\vect{v}$ from above}, denoted $\vect{v}\leq \vect{u}$, if $v_i\leq u_i$ for all $1\leq i\leq n$. 
A vector $\vect{u}\in \Psi$ is \emph{Pareto optimal} if 
$\vect{u}$ is no point in $\Psi$ dominates it from above.
A \emph{Pareto curve} in $\Psi$ is the set of Pareto optimal vectors in $\Psi$. 
The following lemma follows from the \emph{separating hyperplane} and \emph{supporting hyperplane theorems}.
\begin{lemma}[\cite{boyd2004convex}]\label{lem:hyperplane}
    Let $\Psi\subseteq \mathbb{R}^n$ be any downward closure of a convex hull constructed from a set of points $\vect{x} \in \mathbb{R}^{n}$. For any $\vect{v}\not\in \Psi$, there is a weight vector $\vect{w}$ such that $\vect{w}\cdot\vect{v}> \vect{w}\cdot\vect{x}$ for all $\vect{x}\in \Psi$. We say that $\vect{w}$ \emph{separates} $\vect{v}$ from $\Psi$.
    Also, for any $\vect{u}$ on the Pareto curve of $\Psi$, there is a weight vector $\vect{w}'$ such that $\vect{w}'\cdot\vect{u}\geq \vect{w}'\cdot\vect{x}$ for all $\vect{x}\in \Psi$. We say that $\{\vect{x}\in\mathbb{R}^n\mid \vect{w}'\cdot\vect{x}=\vect{w}'\cdot\vect{u}\}$ is a \emph{supporting hyperplane} of $\Psi$.
\end{lemma}

\parahead{Bistochastic Matrix}
For a matrix $\vect{U}\in \mathbb{R}^{n\times n}$ for some $n$, let $u_{i,j}$ denote the element of $\vect{U}$ in the $i^\mathrm{th}$ row and $j^\mathrm{th}$ column. $\vect{U}$ is \emph{bistochastic} if $u_{i,j}\geq 0$ and $\sum_{i'=1}^n u_{i',j}=\sum_{j'=1}^n u_{i,j'}=1$ for all $1\leq i,j\leq n$. A bistochastic matrix $\vect{U}$ is a \emph{permutation matrix} if $\vect{U}$ has exactly one element with value $1$ in each row and each column. We recall the following Birkhoff–von Neumann Theorem:

\begin{lemma}[\cite{birkhoff1946three}]\label{lem:bvn}
A bistochastic matrix $\vect{U}$ 
of order $n$ is equivalent to a convex combination of permutation matrices $\vect{U}_1,\ldots, \vect{U}_k$ for some $k \leq n^2-2n+2$. 
\end{lemma}

\parahead{Random Assignment}
Given a set $I$ (resp., $J$) of agents (resp., tasks) with $|I|= |J|$, a (balanced) \emph{assignment} is a bijective function $f: J\to I$.
Denote the set of assignments of $J$ to $I$ by $\mathcal{F}$. 
A \emph{random assignment} $\dist$ is a randomised distribution over $\mathcal{F}$ (or, equivalently, a convex combination of assignments in $\mathcal{F}$).
For convenience, let $I=J=\{1,\ldots, n\}$. 
Let $\dist_{j\to i}=\dist(\{f\in \mathcal{F}\mid f(j)=i\})$, namely, the marginal probability of assigning task $j$ to agent $i$ according to $\dist$.
Clearly, any assignment is equivalent to a permutation matrix.
Moreover, by Lemma~\ref{lem:bvn} any bistochastic matrix $\vect{U}$ is equivalent to a random assignment $\dist$ such that $u_{i,j}= \dist_{j\to i}$.

\section{Problem and Approach}\label{sec:approach}
\subsection{Problem Statement}

In our MAS setting, each agent is an MDP (with a reward structure) and each task is a DFA (equivalently, a co-safe LTL formula), and the rewards are the probabilities 
of accomplishing the tasks and the costs (as negative rewards) 
of agents executing tasks.
Therefore, we aim to compute a random task assignment and schedulers for all agents and tasks, which 
must meet multiple probability and cost requirements.
Intuitively, we consider the task assignment and agent planning scenario satisfying the following two conditions \cite{schillinger_simultaneous_2018}:
\begin{description}
    \item[C1.] The tasks are mutually independent. 
    \item[C2.] The behaviours of agents do not impact each other.
\end{description}

 
For each $(i,j)\in I\times J$,\footnote{Throughout the paper we assume $I=J=\{1,\ldots,n\}$ for some $n$ (unless explicitly stated otherwise) but still use $I,J$ to indicate the agent or task references.} we define an \emph{agent-task} (product) MDP $\mathcal{M}_{i\otimes j}[\rho_i]\doteq \mathcal{M}_i[\rho_i]\otimes \dfa _j$ 
and include an atomic proposition $\mathtt{done}_j$ such that
\begin{itemize}
    \item[] $L_{i,j}(s,q)\models\mathtt{done}_j$ iff $q\in Q_{j,F}\cup Q_{j,R}$ 
\end{itemize} 
which indicates ``task $j$ is ended (either accomplished or failed).'' 
For each $j\in J$ we define a designated reward function $\rho_{j+|I|}:\bigcupdot_{i\in I} (S_i\times Q_j\times A_i)\to \{0,1\}$ such that $\rho_{j+|I|}(s,q,a)=1$ iff $q\notin Q_{j,F}$ and $\mathrm{suc}(q)\subseteq Q_{j,F}$. If such a pre-sink $q$ does not exist, we can modify $\dfa_j$ to include $q$ without altering $acc(\dfa_j)$. In words, $\rho_{j+|I|}$ provides a one-off unit reward whenever an accepting location will be traversed \emph{for the first time}. Informally, $\rho_{j+|I|}$ expresses ``the probability of accomplishing task $j$.''
As the atomic proposition $\mathtt{done}_j$ is fixed for each $\mdp_{i\otimes j}$, we \emph{abbreviate $\mathbf{E}^{\mathcal{M}_{i\otimes j}[\rho_{k}],\schr_{i,j}}(\evt \mathtt{done}_{j}) $ as $\mathbf{E}^{\mathcal{M}_{i\otimes j}[\rho_{k}],\schr_{i,j}}$} where $k=j$ or $k=j+|I|$. 
As the reachability rewards for agents may be infinite and cause instability in computation, similar to the multi-objective verification literature \cite{forejt2012pareto,hahn2017multi}, we require that $\mdp_{i\otimes j}[\rho_i]$ is reward-finite (w.r.t.\ $\evt\mathtt{done}_j$) for all $(i,j)\in I\times J$. 

\begin{definition}[MORAP] \label{def:MORAP} 
A \emph{multi-objective random assignment and planning} (MORAP) problem is finding a bistochastic matrix $(x_{i,j})_{i\in I,j\in J}$ and a set of schedulers $\{\schr_{i,j}\in Sch(\mathcal{M}_{i\otimes j})\mid i\in I, j\in J\}$ 
such that the following two kinds of requirements, 
namely R1 and R2, are satisfied:
\begin{description}
\item[\textbf{(R1. Probability)}] $\textstyle\sum_{i\in I} x_{i,j}\mathbf{E}^{\mathcal{M}_{i\otimes j}[\rho_{j+|I|}],\schr_{i,j}}\geq p_j$ for all $j\in J$,
\item[\textbf{(R2. Cost)}] $\textstyle\sum_{j\in J} x_{i,j}\mathbf{E}^{\mathcal{M}_{i\otimes j}[\rho_i],\schr_{i,j}}\geq c_i$  for all $i\in I$,
\end{description}
where the probability thresholds $ (p_j)_{j\in J} \in [0,1]^{|J|}$ and the cost thresholds $(c_i)_{i\in I}\in \mathbb{R}^{|I|}$ are given.
If the above requirements are satisfied, we say that the MORAP problem is \emph{feasible} with given thresholds or just that the thresholds are \emph{feasible}.
\end{definition}

Definition~\ref{def:MORAP} is an adequate formulation in the presence of conditions C1 and C2. First, since tasks are mutually independent (C1), the probability requirements only need to address the successful probability of each task. Second, since the execution of any task by each agent does not impact other agents (C2), the cost requirements only need to consider the cost of each agent. In practice, we can relax the condition $|I|=|J|$ to $|I|\geq |J|$ (e.g., adding dummy tasks whose probability threshold is $0$).

\subsection{Convex Characterisation and Centralised Model}\label{sec:convexity}

\begin{figure}[t]
\begin{equation*}
\boxed{
\begin{aligned}
& \text{Maximise}\\
&
\begin{cases}
    \sum_{j\in J}\sum_{(s,q)\in S_i\times Q_j}\sum_{a\in A_i(s)} \rho_i(s,q,a)x_{s,q,a} & \forall i\in I \\
    \sum_{i\in I}\sum_{(s,q)\in S_i\times Q_{j}}\sum_{a\in A_i(s)}\rho_{j+|I|}(s,q,a) x_{s,q,a} & \forall j\in J \\
\end{cases} \\
& \text{Subject to}\ \forall i\in I,j\in J ,  (s, q)\in S_i\times Q_j\text{:}\\ 
& \begin{cases}
    \sum_{a\in A_i(s)}x_{s,q,a} - \mathrm{I}_{(s,q)=(s_{i,0},q_{j,0})}x_{i,j} \\ 
    \quad =\ \sum_{(s',q')\in S_i\times Q_j}\sum_{a'\in A_i(s')}  P_{i,j}(s',q',a',s,q) x_{s',q',a'} \\
    x_{s,q,a} \geq 0;\
    x_{i,j} \geq 0; \
    \sum_{i'\in I} x_{i',j}= 1; \
    \sum_{j'\in I} x_{i,j'}= 1 \\
\end{cases}\\
\end{aligned}
}
\end{equation*}
\caption{The multi-objective linear program for MORAP}
    \label{fig:lp}
    \vspace{-5mm}
\end{figure}

An essential characteristic of our MORAP problem is \emph{convexity}, namely, the downward closure of feasible probability and cost thresholds is a convex polytope (i.e., the downward convex hull of some finite set of points).
This follows from the fact that the MORAP problem can be expressed as a multi-objective linear program (LP) by using a similar technique which underpins multi-objective verification of MDPs \cite{forejt2012pareto,papadimitriou2000approximability,etessami2007multi}. 
Fig.~\ref{fig:lp} includes the multi-objective LP for MORAP.
Intuitively, for each $(i,j)\in I\times J$, $x_{i,j}$ represents the probability of assigning $j$ to $i$ (c.f., Lemma~\ref{lem:bvn}), and for each $(s,q)\in S_i\times Q_i$, $x_{s,q,a}$ is the expected frequency of visiting $(s,q)$ and taking action $a$. A memoryless scheduler can be defined as follows: $\schr_{i,j}(s,q)(a)=x_{s,q,a}/x_{s,q}$ where $x_{s,q}=\sum_{a\in A_i(s)}x_{s,q,a}$.
Thus, a MORAP problem has the following time complexity:
\begin{theorem}\label{thm:complexity}
The feasibility of a MORAP problem is decidable in time polynomial in $\sum_{i\in I,j\in J}|\mdp_{i\otimes j}|$.
\end{theorem}

LP is not efficient for large problems, and value- and policy-iteration methods are more scalable methods in practice. For this purpose, we define a 
new MDP which combines all agent-task MDPs and includes an additional variable indicating which agents have been assigned with tasks.
This MDP is targeted directly at solving the random assignment problem in a centralised way

\begin{definition}[Centralised MDP] \label{def:teamMDP}
A centralisd MDP is $\tmdp  = (S^\tlab , {s}^\tlab _{0}, {A}^\tlab , $ ${P}^\tlab , {L}^\tlab )$
where (i) $S^\tlab = \bigcupdot_{i\in I}\bigcupdot_{j\in J} S_{i}\times Q_{j} \times 2^{I}$, (ii) $s^\tlab _{0} = (s_{1,0}, q_{1,0},\emptyset)$, (iii) $A^\tlab =\bigcup_{i\in I} A_{i} \cupdot \{b_1,b_2,b_3\}$, (iv) $P^\tlab  = S^\tlab \times A^\tlab \times S^\tlab \to [0,1]$ such that:
\begin{itemize}
		\item ${P}^\tlab  (s, q,\sharp  , a,s',q', \sharp) ={P}_{i,j}  (s, q,a,s',q')$ if $s,s'\in S_i$, $q,q'\in Q_j$, $a\in A_i(s)$ and $i\in \sharp$ for some $i,j$,
 		\item ${P}^\tlab  (s, q,\sharp, b_1,s, q, \sharp\cup \{i\}) = 1$  if $s = s_{i,0}$, $q = q_{j,0}$ and $i\notin \sharp$ for some $i,j$,
 		\item ${P}^\tlab  (s, q,\sharp, b_2,s', q, \sharp) = 1$  if $s = s_{i,0}$, $q = q_{j,0}$ and $s'=s_{i',0}$ with $i'=\min\{i''\in I \mid i''>i, i''\notin \sharp\}$ for some $i,j$,
 		\item ${P}^\tlab  (s,q,\sharp, b_3,s', q',\sharp) = 1$  if $s \in  S_i$, $q \in Q_{j,F}\cup Q_{j,R}$, $i
 		\in \sharp$, $s'=s_{i',0}$ with $i'= \min\{i''\in I \mid i''\notin \sharp\}$, and $q'= q_{j+1,0}$ for some $i$, $j<|J|$,
 	\end{itemize}
(v) $L^\tlab : S^\tlab \to 2^{\{ \mathtt{done}\}}$ such that $ L^\tlab (s,q)\models \mathtt{done}$ iff $q\in Q_{j,F}\cup Q_{j,R}$.
\end{definition}

\noindent Intuitively, $\sharp
\subseteq I$ indicates agents who have worked on some tasks; $b_1$ indicates ``a task is assigned to the current agent''; $b_2$ indicates ``a task is forwarded to the next agent''; and $b_3$ indicates ``the next task is considered''. The model behaves as an individual product MDP when working on the assigned tasks.

Given any reward structure $\rho$ for $\mdp_{i\otimes j}$, we view $\rho$ as a reward structure for $\tmdp$ by letting $\rho(s,q,\sharp,a)=\rho(s,q,a)$ 
if $a \in A_i(s)$ 
and $\rho(s,q,\sharp,a)=0$ otherwise for all $(s,q,\sharp,a)$. Similarly, given any reward structure $\rho$ for $\tmdp$, a restriction of $\rho$ on $S_i\times Q_j\times A_i$ is a reward structure for $\mdp_{i\otimes j}$. Similar to agent-task MDPs, we abbreviate $\mathbf{E}^{\tmdp [\rho],\schr}(\evt \mathtt{done})$ as $\mathbf{E}^{\tmdp [\rho],\schr}$ for a given $\rho$.

\begin{theorem} \label{thm:main}
The MORAP problem in Definition~\ref{def:MORAP} is feasible with respect to $(p_j)_{j\in J}$ and $(c_i)_{i\in I}$ 
if and only if there is $\schr\in Sch(\tmdp )$ such that $\mathbf{E}^{\tmdp [\rho_{j+|J|}],\schr}\geq p_j$ and $\mathbf{E}^{\tmdp [\rho_{i}],\schr}\geq c_i$ for all $i\in I,j\in J$.
\end{theorem}

With the above theorem, one can work on the centralised MDP $\tmdp$ (e.g., by using value-iteration) to solve a MORAP problem. Therefore, existing probabilistic model checking tools for multi-objective MDP verification (e.g., Prism \cite{KNP11} and Storm \cite{hensel2022probabilistic}) can be employed.
However, the state space of $\tmdp$ is exponential with respect to the agent team size $|I|$. Therefore, this approach is hard to scale to a relatively large $|I|$.

\subsection{Point-Oriented Pareto Computation by Decentralised Model}
\label{sec:decentralised-synthesis}
We present a decentralised method solve a given MORAP problem, especially when the agent number (i.e., task number) is large.
Besides deciding whether the problem is feasible or not, for a non-feasible problem our method also computes a new feasible threshold vector on the Pareto curve of the problem, and nearest the original threshold vector up to some numerical tolerance.
\begin{algorithm}[t]
\caption{Point-oriented Pareto computation\label{alg:scheduler-synthesis}}
\SetKwFor{ParallelForEach}{parallel foreach}{do}{}
\SetKwBlock{DoParallel}{do in parallel}{}
\KwIn{$\{\mathcal{M}_{i\otimes j}\}_{(i,j)\in I\times J}$, $\vect{\rho}=\{\rho_k\}_{k=1}^{|I|+|J|}$,
$ \vect{t}$ (a concatenation of $\vect{c}$ and $\vect{p}$), $\varepsilon\geq 0$}
$\vect{t}_{\uparrow}:=-\vect{\infty}$; $\vect{t}_{\downarrow}:=\vect{t}$; $\Phi:=\emptyset$; $\Lambda: =\emptyset$; $\vect{w}:=(1,0,\ldots,0)$\;
\While{$\|\vect{t}_{\downarrow}- \vect{t}_{\uparrow}\|>\varepsilon$ \label{ln:condition0}}{
\If{$\Phi\neq \emptyset$}{
\label{ln:inner_start}
Find $\vect{x}\in {down}(\Phi)$ minimising $\|\vect{t}-\vect{x}\|$\;\label{ln:find-minimum1}
$\vect{t}_{\uparrow}:= \vect{x}$\; 
$\vect{w}:=\vect{M}(\vect{t}-\vect{t}_{\uparrow})/\|\vect{M}(\vect{t}-\vect{t}_{\uparrow})\|_1$\;\label{ln:find-w-1}
}
Find $\vect{r}$ s.t.\ $\{\vect{y}\mid \vect{w}\cdot\vect{y}= \vect{w}\cdot\vect{r}\}$ is a supporting hyperplane of $\mathscr{C}$ \;\label{ln:find-supp-hp}
$\Phi:=\Phi\cup \{\vect{r}\}$; $\Lambda:=\Lambda\cup \{(\vect{w},\vect{r})\}$\;\label{ln:add-vector}
\If{$  \vect{w} \cdot\vect{r}  <  \vect{w} \cdot \vect{t}_{\downarrow}$\label{ln:condition1}}{
Find $\vect{z}$ minimising $\|\vect{t}-\vect{z}\|$ s.t.\ $\vect{w}'\cdot\vect{r}' \geq \vect{w}'\cdot \vect{z}$ for all $(\vect{w}',\vect{r}')\in \Lambda$\;\label{ln:find-minimum2}
$\vect{t}_{\downarrow}:= \vect{z}$\;\label{ln:update-t-tmp}
}
}
\end{algorithm}


Let $\mathscr{C}_0= \{ (\mathbf{E}^{\tmdp[\rho_k],\schr})_{1\leq k\leq |I|+|J|}\mid \schr\in {Sch}(\tmdp)\}$. The reward-finiteness implies that $\mathscr{C}_0$ is non-empty and bounded. Let $\mathscr{C}$ be the downward closure of $\mathscr{C}_0$, i.e., 
namely, $\mathscr{C}$ is \emph{the set of feasible threshold vectors} in Definition~\ref{def:MORAP}.
The main algorithm for our method is presented in Algorithm~\ref{alg:scheduler-synthesis} with the supporting hyperplane computation (i.e., Line~\ref{ln:find-supp-hp}) detailed in Algorithm~\ref{alg:find-supporting-hyperplane}.
Algorithm~\ref{alg:scheduler-synthesis} works by iteratively refining a \emph{lower approximation}, encoded as $\Phi$, and an \emph{upper approximation}, encoded as $\Lambda$, for $\mathscr{C}$. It computes a vector $\vect{t}_{\uparrow}$ (resp., $\vect{t}_{\downarrow}$) which is the closest point from the origin threshold vector $\vect{t}$ to the lower (resp., upper) approximation such that $\vect{t}_{\uparrow}$ and  $\vect{t}_{\downarrow}$ converge eventually. 

The algorithm uses a general norm $\|\cdot\|$ to measure the distance between vectors, as in practice one may prefer to differentiate the importance of probability and cost thresholds.
An \emph{inner product} of $\vect{v},\vect{u}\in \mathbb{R}^{m}$ ($m$ a positive integer), denoted $\langle\vect{v},\vect{u}\rangle$,  is the matrix-vector multiplication $\vect{v}^T\vect{M}\vect{u}$, where $\vect{M}$ is a symmetric positive-definite matrix. 
Note that if $\vect{M}$ is the identity matrix then $\langle\vect{v},\vect{u}\rangle$ is $\vect{v}\cdot\vect{u}$.
Then, $\|\vect{v}|=\langle\vect{v},\vect{v}\rangle$.
Let $\|\cdot\|_1$ denote vector 1-norm. 
The weight vector $\vect{w}$ computed in Line~\ref{ln:find-w-1} is the (opposite) direction of projecting $\vect{t}$ onto ${down}(\Phi)$. 
Moreover, 
$\vect{w}\cdot\vect{\rho}$ denotes a weighted combination of reward functions in $\vect{\rho}$.

\begin{theorem}\label{thm:synthesis}
Algorithm~\ref{alg:scheduler-synthesis} terminates for any $\varepsilon\geq 0$. Throughout the execution of Algorithm~\ref{alg:scheduler-synthesis}, the following properties hold: (i) $\vect{t}_{\uparrow}\in \mathscr{C}$. (ii) If $\vect{t}\in \mathscr{C}$ then $\vect{t}_{\downarrow} = \vect{t}$. (iii) $\|\vect{t}-\vect{t}_{\downarrow}\|\leq \min_{\vect{u}\in\mathscr{C}}\|\vect{t} -\vect{u}\|\leq \|\vect{t}-\vect{t}_{\uparrow}\|$.
\end{theorem}
\begin{corollary}\label{cor:}
Suppose $\varepsilon=0$. After Algorithm~\ref{alg:scheduler-synthesis} terminates, the following properties hold: (i) $\vect{t}_{\uparrow}=\vect{t}_{\downarrow}$. (ii) $\vect{t}\in \mathscr{C}$ if and only if $\vect{t}_{\downarrow} = \vect{t}$. (iii) If $\vect{t}\notin \mathscr{C}$ then $\vect{t}_{\downarrow}$ is on the Pareto curve of $\mathscr{C}$.
\end{corollary}

\begin{algorithm}[t]
\SetKwFor{ParallelForEach}{parallel foreach}{do}{}
\SetKwBlock{DoParallel}{do in parallel}{}
\KwIn{$\{\mathcal{M}_{i\otimes j}\}_{(i,j)\in I\times J}$, $\vect{\rho}=\{\rho_k\}_{k=1}^{|I|+|J|}$, $\vect{w}$}
 \ForEach{$(i,j)\in I\times J$}{
\tcc{Line~\ref{ln:find-scheduler-1} is computed by policy iteration.}
\label{ln:loop-1-start}
$c_{i,j}:= \mathbf{E}^{\mathcal{M}_{i\otimes j}[  \vect{w} \cdot \vect{\rho}],\schr_{i,j}}$ with $\schr_{i,j}:=\argmax_{\schr} \mathbf{E}^{\mathcal{M}_{i\otimes j}[  \vect{w} \cdot \vect{\rho}],\schr}$\;\label{ln:find-scheduler-1}
}
Find an assignment $f\in \mathcal{F}$ maximising $ \sum_{j\in J} c_{f(j),j}$\;\label{ln:find-assignment}
\ForEach{$j\in J$}{
\tcc{Lines~\ref{ln:verification-1}-\ref{ln:verification-2} are computed by value iteration.}
\label{ln:loop-2-start}
$r_{j+|I|}:= \mathbf{E}^{\mathcal{M}_{f(j)\otimes j}[\rho_{j+|I|}],\schr_{f(j),j}}$\;\label{ln:verification-1}
$r_{f(j)}:= \mathbf{E}^{\mathcal{M}_{f(j)\otimes j}[\rho_{f(j)}],\schr_{f(j),j}}$\;\label{ln:verification-2}
}
\textbf{return} $(r_k)_{k=1}^{|I|+|J|}$\;  
\caption{Supporting hyperplane computation in Line~\ref{ln:find-supp-hp} of Alg.~\ref{alg:scheduler-synthesis}\label{alg:find-supporting-hyperplane}}
\end{algorithm}
%

Algorithm~\ref{alg:find-supporting-hyperplane} finds a supporting hyperplane of $\mathscr{C}$ for a given orientation $\vect{w}$. As probabilistic model checking 
is employed in the two inner loops, it is usually an expensive computation.
To see the significance of Algorithm~\ref{alg:find-supporting-hyperplane}, we point out that $\mathscr{C}$ is a convex set defined on the centralised model $\tmdp$ whose size is $O(2^{|I|})$. But instead of dealing with $\tmdp$, Algorithm~\ref{alg:find-supporting-hyperplane} works on a decentralised model consisting of $\{\mathcal{M}_{i\otimes j}\}_{(i,j)\in I\times J}$. 
The first inner loop includes $|I|\times |J|$ (i.e., $|I|^2$) policy-iteration processes to compute optimal schedulers and reachability rewards. The second inner loop uses $2|I|$ value-iteration processes under a fixed scheduler.\footnote{The methods for computing the two inner loops are detailed in the appendix of \cite{long_version}.}
The model selection is computed by using the Hungarian algorithm \cite{kuhn1955hungarian} (Line~\ref{ln:find-assignment}) whose run time is $O(|I|^3)$.
Another important implication of using a decentralised model is the parallel execution of the two inner loops, which we elaborate on in Section~\ref{sec:implementation}. 
Also notice that if $\mathcal{M}_{i \otimes j}= \mathcal{M}_{i' \otimes j'}$ for some $(i,j) \neq (i',j')$, then some models can be skipped in the two inner loops.
In the implementation we should choose some positive $\varepsilon$ for the following three reasons: First, the policy and value iterations for computing the two inner loops are approximate. 
Second, small numerical inaccuracy (e.g., rounding) usually 
occurs in the solving optimisation problems in the algorithm. Third, as the worst-case number of iterations in Algorithm~\ref{alg:scheduler-synthesis} is exponential on the model size and agent number \cite{forejt2012pareto}, a suitable $\varepsilon$ can terminate the algorithm earlier with an approximate threshold vector whose precision is acceptable in practice. 

For synthesis purposes, we can extract a random assignment and a collection of schedulers. 
Assume that the while loop iterates $\ell$ times in total. Let $\{\schr_{i,j}^{\iota}\mid i\in I, j\in J, \}$ and $f_{\iota}$ be generated in Lines~\ref{ln:find-scheduler-1}-\ref{ln:find-assignment} in Algorithm~\ref{alg:find-supporting-hyperplane}, respectively, in the $\iota^\mathrm{th}$ iteration. Let $
v_1 \vect{r}_1 +\ldots +v_\ell \vect{r}_\ell\geq \vect{t}_{\uparrow}$ for some weight vector $\vect{v}$ (this $\vect{v}$ exists since $\vect{t}_{\uparrow}\in{down}(\Phi)$). The convex combination of assignments $v_1 f_1+\ldots +v_\ell f_\ell$ defines a random assignment (i.e., bistochastic matrix).
After an assignment $f_\iota$ is chosen randomly according to probability $v_\iota$, the schedulers for planning are those from  $\{\schr_{i,j}^{\iota}\mid  j\in J, f_{\iota}(j) = i \}$.
\tikzstyle{basic} = [rectangle, minimum width=1.2cm, minimum height=1cm, align=center, draw=black]

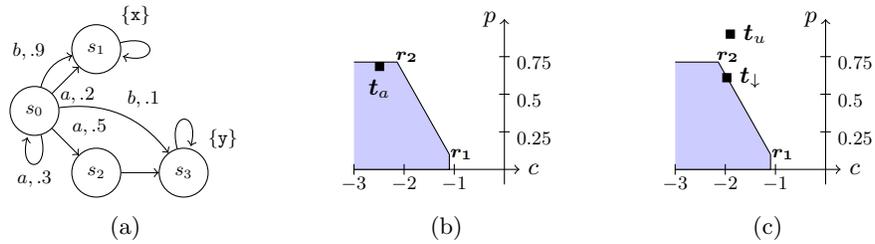
\begin{figure}[t]
    \centering
    \begin{subfigure}[b]{0.3\textwidth}
         \centering
         \begin{tikzpicture}[scale=0.6, every node/.style={scale=0.8}, node distance=0.5cm, auto]
         \node (s_0) [basic, state] {$s_0$};
         \node[label=above right:{$\{\mathtt{x}\}$}] (s_1) [state, above right = of s_0] {$s_1$};
         \node (s_2) [state, below right = of s_0] {$s_2$};
         \node[label=above right:{$\{\mathtt{y}\}$}] (s_3) [state, right = of s_2] { $s_3$};
         \path[->] (s_0) edge [loop below] node {$a,.3$} (s_0)
                   (s_0) edge node [below, xshift=0.2cm] {$a, .2$} (s_1)
                   (s_0) edge node {$a, .5$} (s_2)
                   (s_0) edge [bend left] node {$b, .9$} (s_1)
                   (s_0) edge [bend left] node {$b, .1$} (s_3)
                   (s_2) edge node {} (s_3)
                   (s_3) edge [loop above] (s_3)
                   (s_1) edge [loop right] (s_1);
         \end{tikzpicture}
         \caption{}
         \label{fig:ex_automata}
     \end{subfigure}
     \hfill
     \begin{subfigure}[b]{0.3\textwidth}
         \centering
         \begin{tikzpicture}
             \draw[->, scale=2.] (-1.0, 0) -- (0.1, 0) node[right] {$c$};
             \draw[->, scale=2.] (0, -0.1) -- (0, 1.0) node[left] {$p$};
             \draw[scale=2., fill=blue!20] plot coordinates { (-1, 0.) (-0.366, 0.0) (-0.366, 0.1) (-0.7133, 0.714) (-1., 0.714)};
             \node at (-0.57, 0.2) {\scriptsize$\vect{r_1}$};
             \node at (-1.3, 1.5) {\scriptsize$\vect{r_2}$};
             \foreach \x in {-1, ..., -3} {%
                \draw ($(\x * 2 / 3,0) + (0,-\TickSize)$) -- ($(\x *  2 / 3,0) + (0,\TickSize)$) node[scale=0.8, below, yshift=-0.1cm] {$\x$};
             }
             \foreach \y in {0.25, 0.5, 0.75} {%
                \draw ($(0, \y * 2) + (-\TickSize, 0)$) -- ($(0, \y * 2) + (\TickSize, 0)$) node[scale=0.8, right, yshift=-0.1cm] {$\y$};
             }
             \node[draw, scale=0.5, fill, label=below:{$\vect{t}_a$}] at (-1.66, 1.37) {};
         \end{tikzpicture}
         \caption{}
         \label{fig:ex_achievable}
     \end{subfigure}
     \hfill
     \begin{subfigure}[b]{0.3\textwidth}
         \centering
         \begin{tikzpicture}
             \draw[->, scale=2.] (-1.0, 0) -- (0.1, 0) node[right] {$c$};
             \draw[->, scale=2.] (0, -0.1) -- (0, 1.0) node[left] {$p$};
             \draw[scale=2., fill=blue!20] plot coordinates { (-1, 0.) (-0.366, 0.0) (-0.366, 0.1) (-0.7133, 0.714) (-1., 0.714)};
             \node at (-0.57, 0.2) {\scriptsize$\vect{r_1}$};
             \node at (-1.3, 1.5) {\scriptsize$\vect{r_2}$};
             \node[draw, scale=0.5, fill, label=right:{$\vect{t}_u$}] at (-1.266, 1.8) {};
             \node[draw, scale=0.5, fill, label=right:{$\vect{t}_\downarrow$}] at  (-1.313, 1.22) {};
             \foreach \x in {-1, ..., -3} {%
                \draw ($(\x * 2 / 3,0) + (0,-\TickSize)$) -- ($(\x *  2 / 3,0) + (0,\TickSize)$) node[scale=0.8, below, yshift=-0.1cm] {$\x$};
             }
             \foreach \y in {0.25, 0.5, 0.75} {%
                \draw ($(0, \y * 2) + (-\TickSize, 0)$) -- ($(0, \y * 2) + (\TickSize, 0)$) node[scale=0.8, right, yshift=-0.1cm] {$\y$};
             }
         \end{tikzpicture}
         \caption{}
         \label{fig:ex_unachievable}
     \end{subfigure}
        \caption{Example MOMDP agent and corresponding execution of Algorithm \ref{alg:scheduler-synthesis}.}
        \label{fig:three graphs}
        \vspace{-5mm}
\end{figure}
\parahead{Example} Fig.~\ref{fig:three graphs} is a simple example consisting of one agent and one task to demonstrate an execution of Algorithm \ref{alg:scheduler-synthesis}. Fig.~\ref{fig:ex_automata} shows the agent MDP, where $\rho(s,a) = -1$ for each $a\in A(s)$ and $s \in S$, and the task is $\varphi:=\neg \mathtt{x} \utl \mathtt{y}$. Let $\varepsilon=0.001$.
Fig.~\ref{fig:ex_achievable} shows the computation with a feasible threshold vector $\vect{t}_a = (-2.5, 0.7)$. 
Initially, $\vect{w} = (1, 0)$, which results in $\vect{r}_1=(-1.1, 0.1)$ and the hyperplane $\vect{w}\cdot\vect{x} = \vect{w}\cdot\vect{r}_1= -1.1$. Here, $\|\vect{t}_\downarrow - \vect{t}_\uparrow\| = 0.6$ and so another iteration is needed. 
The algorithm finds $\vect{w}=(0.4, 0.6)$ and the corresponding $\vect{r}_2 = (-2.1, 0.71)$. As $\vect{t}_a$ is contained in $down(\{\vect{r}_1, \vect{r}_2\})$, the algorithm terminates. 
Fig.~\ref{fig:ex_unachievable} shows the case with a non-feasible $\vect{t}_u = (-1.8, 0.9)$. Similar to the previous case, the algorithm finds $\vect{w} = (1, 0)$ and $\vect{r}_1$, and then $\vect{w} = (0.4, 0.6)$ and  $\vect{r}_2$.  
As $\vect{w} \cdot \vect{r}_2 < \vect{w} \cdot \vect{t}_u$, 
it finds a new threshold vector $\vect{t}_\downarrow = (-1.97, 0.61)$ in Line~\ref{ln:find-minimum2}. 
Now as $\|\vect{t}_\downarrow - \vect{t}_\uparrow\| < \varepsilon$, the algorithm terminates.
\section{Hybrid GPU-CPU Implementation}
\label{sec:implementation}

In modern systems, GPU and multi-core CPU hardware is readily available. 
We developed an implementation for our MORAP framework, which utilises heterogeneous GPU and multi-core CPU resources to accelerate the computation. The acceleration is based on non-shared data within the two probabilistic model checking loops in Algorithm~\ref{alg:find-supporting-hyperplane}, which takes up the majority of run time for Algorithm~\ref{alg:scheduler-synthesis} in practice.
Parallel execution on GPU and CPU is by allocating models to each available GPU device and CPU core.
For GPU, further (massive) parallelisation can be achieved on the low-level matrix operations for probabilistic model checking.

\parahead{Implementation goal} The main goal of our framework is to maximise throughput and parallelism.
Combination of multiple devices is a load balancing problem in which we can effectively schedule model checking problems to keep all devices optimally busy, and reduce run time. 
We say that computations run on GPU are called \emph{device} operations. 
A multi-core processor can leverage shared memory with negligible latency before computing.
The main concern with parallelism when using a multi-core processor is \emph{thread-blocking} and \emph{context switching} overhead which should be avoided.
Moreover, because low level computations are sequential a processor's execution run time will correspond to the size of a model's state space. 
On the other hand, the major issue with computing on GPU is data transfer between the host and the device. 

\begin{figure}[t]
\centering
\includegraphics[width=0.7\textwidth]{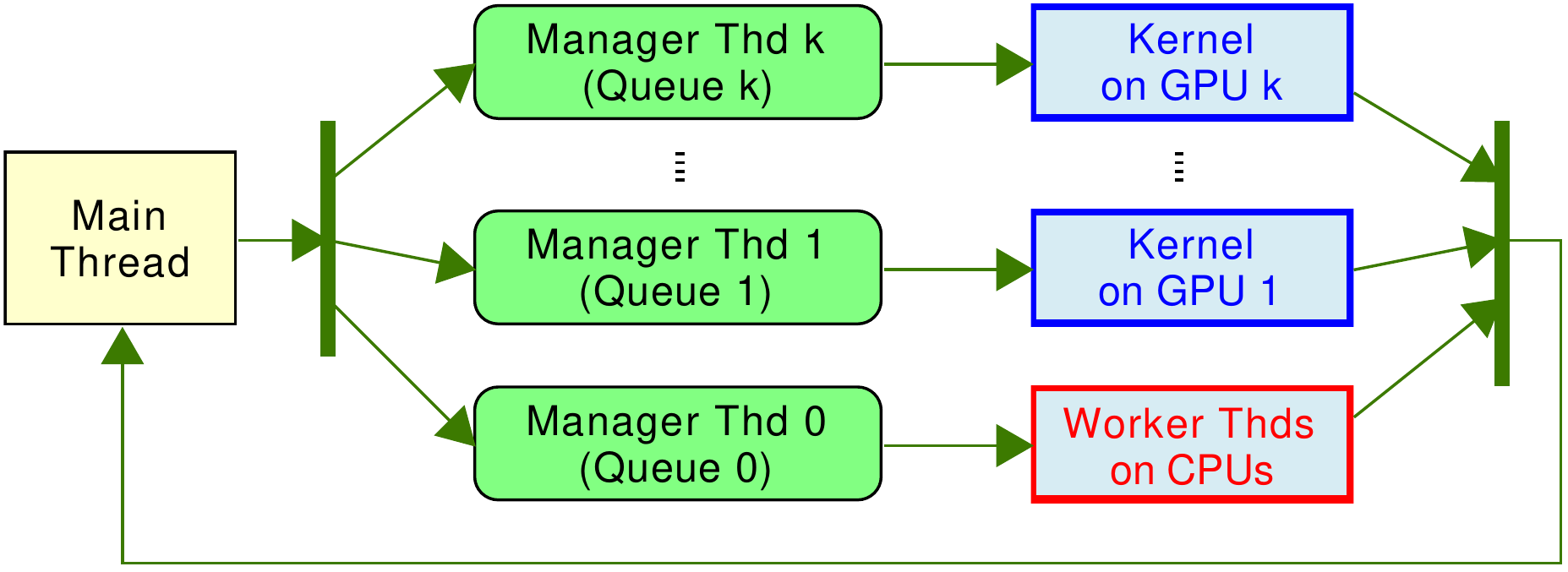}
\caption{Parallel architecture of MORAP framework.} 
\label{fig:arch}
\end{figure}

\renewcommand{\arraystretch}{1.1}
\begin{table}[t]
\vspace{-5mm}
    \caption{Thread roles in MORAP implementation\label{tab:my_label}}
    \begin{center}   
    \begin{tabular}{|@{\hspace{.5em}}l@{\hspace{.5em}}|@{\hspace{.5em}}p{0.7\textwidth}@{\hspace{.5em}}|} \hline
        \textbf{Component} & \textbf{Functionality} \\ \hline
        Main Thread & Loading models to the main memory and running all computation except the two (inner) loops in Alg.~2; generating and allocating models to queues for worker threads and kernels \\ \hline
        Manager Thread & Managing the bounded FIFO queues; (one thread) spawning CPU Worker Threads; (the other threads) calling GPU kernels, incl.\ copying data between the host memory and GPU device memory; communicating with each other via a messaging channel for load balancing  \\ \hline
        Worker Thread & Computing the loops in Alg.~2; each thread bounded on one CPU core and handling one model each time\\ \hline
        Kernel on GPU & Computing the loops in Alg.~2; each kernel running on one GPU device and handling one model each time\\ \hline
    \end{tabular}
    \end{center}
    \vspace{-5mm}
\end{table}
\renewcommand{\arraystretch}{1}

\parahead{Design} Fig.~\ref{fig:arch} shows the parallel architecture of our framework, and Table~\ref{tab:my_label} explains the roles of thread types. 
In particular, there are $k+1$ manager threads controlling $k+1$ FIFO queues of agent-task models $\mdp_{i\otimes j}$, where $k$ is the number of available GPU devices.
One particular manager thread is 
responsible for spawning {worker} threads bound on each available CPU-core, while the others 
call 
kernel functions on 
the 
GPU.
As each worker thread is dedicated to computing one $\mdp_{i\otimes j}$, response time and {context switching} overhead are minimised. 
Manager threads are not required to be bound to any CPU core as program management is not demanding. 
The computation workload between GPU devices and CPU cores is controlled through a \emph{work stealing}  approach \cite{blumofe1999scheduling}, that is,
if a processor or device is idle and its queue is empty then its manager thread will request (i.e., \emph{steal}) a model from another queue. 
In this way, hardware is optimally loaded with work, and all threads operate asynchronously.

\parahead{Programming and data structure} We implemented our framework with multiple languages including Rust (framework API), CUDA C (GPU device control), and a Python user interface, where the Rust API calls to C, and Python via a foreign function interface (FFI).
Our implementation uses the \emph{affine} property \cite{pierce2004advanced} of the type system in Rust \cite{matsakis2014rust} to ensure that {owned} variables can be used at most once in the application with move-only types. 
This feature is particularly useful in a parallel architecture, as $\mdp_{i \otimes j}$ can be owned by at most one thread at a time, 
and thread computation side-effects are inconsequential to any other thread.
Isolating data access to each $\mathcal{M}_{i \otimes j}$ mitigates the requirement of shared memory access, freeing the framework from data races and data starving. Consequently, the problem is \emph{embarrassingly} parallel. 
Constructing the architecture in this way ensures that our implementation approaches the upper-bound of parallelism. 
Our implementation uses explicit data structures (i.e., sparse matrices) to store the transition probability function and reward structures for each $\mdp_{i\otimes j}$.
Parallel low-level matrix operations on GPU are implemented using the CUDA cuSPARSE API, which guarantees thread-safety.
The reduce operation of state-action values for finding an optimal policy in Line~\ref{ln:find-scheduler-1} is also computed in parallel with one kernel launch.
Optimal occupancy for a GPU kernel is managed through a kernel launcher and a call to CUDA $\mathtt{cudaOccupancyMaxPotentialBlockSize}$.

\section{Experiments}
\label{sec:experiments}

\begin{wrapfigure}{r}{0.45\textwidth}
\vspace{-0.7cm}
\centering
\includegraphics[width=0.9\linewidth]{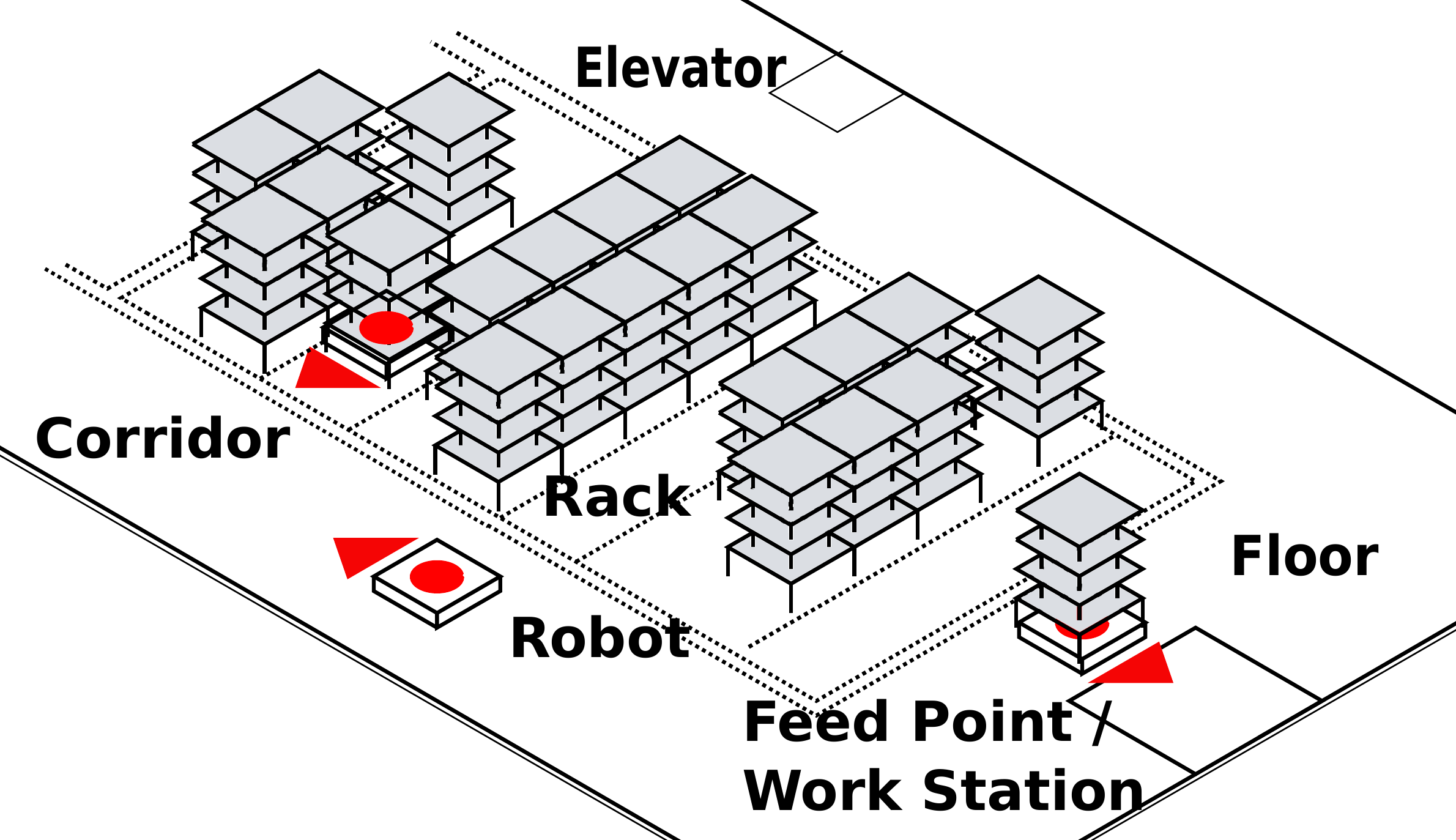}
\caption{A smart warehouse layout}
\vspace{-0.2cm}
\label{fig:warehouse}
\end{wrapfigure}


One realistic example for our MORAP problem is a smart-warehousing or robotic mobile fulfilment system (RMFS), which usually controls tasks centrally with limited communication between robots \cite{wurman2008coordinating}. 
The environment, as depicted in Figure~\ref{fig:warehouse}, is a $W\times H$ two-dimensional grid typically consisting of movable racks (shelves), storage locations, and workstations where order picking and replenishment can take place \cite{merschformann2018rawsim}. 
Robots maneuver in the warehouse to carry out tasks such as order picking and replenishment.

The state of robots is described by the robot position, the internal robot state (e.g., carrying a rack or not) and the environment parameters (e.g., the rack locations) and is discrete. 
Robots can perform such actions as Rotate Left/Right, Go Forward, Load/Unload Rack.
The MORAP problem in this example is (random) assignment $n$ tasks to $n$ robots, and task planning for robots, under the multi-objective requirements of running costs and task fulfilment probabilities.
%
We considered \emph{replenishment tasks} for agents, which are informally described as follows: ``While not carrying anything, go to a rack position in the warehouse, get the rack and carry it to the feed for replenishment, then carry the rack and drop back at a specific rack position.'' Formally, each replenishment task can be specified as a co-safe LTL formula or as a DFA. 
Other tasks such as picking tasks can be specified in a similar way. 

We conducted two experiments to evaluated our MORAP implementation
using Algorithm \ref{alg:scheduler-synthesis}
in our smart warehousing example with different warehouse dimensions $W\times H$ and different agent (task) numbers $n$.
Notice that we evaluated run time per iteration (rather than its end-to-end run time) for our main algorithm (i.e., Algorithm~\ref{alg:scheduler-synthesis}).
\textbf{Experiment 1} included two comparisons 
: First, it compared the model size of the centralised and decentralised models. Second, it compared the run time of hybrid GPU-CPU, (pure) GPU, multi-core CPU and single-core CPU computation.
Note that the hybrid GPU-CPU and multi-core CPU computation is applicable to the decentralised model only.
To benchmark the performance of our implementation with the probabilistic model checking tools Prism and Storm which do not support task assignment problems, \textbf{Experiment 2} compared 
the verification-only average run time for the centralised model for 
our implementation, against Prism and Storm. 
Prism, and Storm work in a similar way to Algorithm  \ref{alg:scheduler-synthesis} by iteratively generating a weight vector $\vect{w}$ and computing a new Pareto point.
In all cases which we had performed, the number of iterations ranged from 2 to 16.


All experiments were conducted on Debian with an AMD 2970WX 24 Core 3.0GHz Processor PCIe 3.0 32Gbps bandwidth, 3070Ti 1.77GHz 8Gb 6144 CUDA Cores GPU, and 32Gb of RAM. An artefact to reproduce the experiments is available online\footnote{https://github.com/tmrob2/hybrid-motap}. 
A single GPU was used and therefore $k=1$ for the number of GPU management threads. 
Prism configuration included using explicit data structures, the Java heap size and hybrid $\mathtt{maxmem}$ were set to 32Gb to avoid memory exceptions.
The default configuration was sufficient for Storm.
The value iteration stopping threshold was set
to $10^{-6}$. Running time cut-off was set to 180 seconds, if the run time exceeds the cut-off time a $\mathtt{timeout}$ error was recorded. If the GPU device runs out of memory, a $\mathtt{memerr}$ was recorded. 
The 
Pareto curve threshold $\varepsilon$ (see line \ref{ln:condition0} in Alg. \ref{alg:scheduler-synthesis})
was set to 0.01.

\begin{table}[t]
\caption{Evaluation of average run time (sec.) 
per iteration in Algorithm~\ref{alg:scheduler-synthesis} for centralised and decentralised models,
where states and transitions refer to reachable states and reachable transitions, respectively.
}
\label{tab:exp1}
\renewcommand{\arraystretch}{1.2}
\centering
\scriptsize
\medskip
\begin{tabular}{|c|c|cc|cccc|cc|cc|}
\hline
\multirow{4}{*}{\shortstack{W.H.\\ size \\$W\times H$}} &
\multirow{4}{*}{\shortstack{agent\\ (task)\\ num.\\ $n$}} &
\multicolumn{6}{c|}{Decentralised} &
\multicolumn{4}{c|}{Centralised} \\
\cline{3-12}
 &
 &
\multicolumn{2}{c|}{Dec. } &
\multicolumn{4}{c|}{Time per iter.} &
\multicolumn{2}{c|}{Cent.} &
\multicolumn{2}{c|}{Time per iter.}
\\
 & &\multicolumn{2}{c|}{Model Size} & Hybrid & Mult.  &GPU & CPU & \multicolumn{2}{c|}{Model Size} & CPU & GPU\\
 & & states & trans. & & CPU & & & states & trans. & & \\
\hline
\multirow{4}{*}{6$\times$6} & 2 & 17k & 104k & 0.016 & \textbf{0.01} & 0.037 & 0.025 &21.2K & 136K & 0.059 & 0.017\\
&5 &106k & 652k & 0.023 & \textbf{0.02} & 0.2 & 0.36 & 3.5M & 22.5M & 6.1 & 2.35 \\
& 6 &152K & 940K & 0.03 & \textbf{0.022} & 0.96 & 0.38 &24.9M & 162M & \scriptsize $\mathtt{timeout}$ & 15.2 \\
& 50 &10.6M & 65.3M & 1.36 & \textbf{1.0} & 27.2 & 11.1 & \scriptsize $\mathtt{memerr}$ & \scriptsize $\mathtt{memerr}$ &  - & - \\ 
& 100 &42.4M & 261M & 4.8 & \textbf{3.9} & 90.8 & 31.9  & - & -  & - & - \\  
\hline
\multirow{8}{*}{12$\times$12}& 2 &254k & 1.5M & 0.18 & 0.14 & 0.13 & \textbf{0.09} & 190k & 1.2M & 1.08 & 1.5 \\
& 4 &1.0M & 6.1M & \textbf{0.36} & 0.38 & 0.33 & 1.78 & 635k & 4.2M & 9.8 & 2.1 \\
& 6 &2.2M & 13.8M & \textbf{0.7} & 0.9 & 0.7 & 4.0 & \scriptsize $\mathtt{memerr}$ & \scriptsize $\mathtt{memerr}$ & - & - \\
& 8 &4.1M & 24.5M & \textbf{1.1} & 1.6 & 1.2 & 7.2 & - & - & - & - \\
& 10 &6.4M & 38.3M & \textbf{1.8} & 4.23 & 2.46 & 11.7 &  & - & - & - \\
& 20 &25.4M & 153M & \textbf{6.5} & 17.3 & 9.8 & 45.7 & - & - & - & - \\ 
& 30 &57.2M & 345M & \textbf{15.3} & 38.8 & 22.1 & $\mathtt{timeout}$ & - & - & - & -\\ 
\hline
\end{tabular}
\newline
\renewcommand{\arraystretch}{1}
\vspace{-5mm}
\end{table}
%

The results for Experiment 1 are included in Table \ref{tab:exp1}.
It can be observed that, in general, the run time performance of the decentralised model is significantly improved over the centralised model.
As expected, the centralised model run time grows exponentially with the increment on the agent and task numbers, while the growth for the decentralised model is linear. 
Table~\ref{tab:exp1} also shows that parallel implementation of some form achieved improved run time performance.  
For a 6$\times$6 warehouse size, multiple CPU achieved almost 10 times improvement over single-CPU. 
For a 12$\times$12 warehouse size, the hybrid GPU-CPU achieved a similar performance increase.
%
When conducting this experiment, we observed that one performance indicator is the ratio of model checking time to model (data) copying time: A higher (resp., lower) ratio implies more (resp., less) effective GPU acceleration. 
This ratio was higher in a 12$\times$12 warehouse than in a 6$\times$6 warehouse, as the former size led larger individual agent-task models than the latter size.
In particular, we observed that a high ratio is important to the hybrid GPU-CPU approach. 
For larger individual agent-task models, the hybrid approach achieved significant improvement over both pure GPU acceleration and multi-core CPU acceleration.

\begin{table}[t]
\centering
\caption{Comparison of verification-only average run time (sec.) per iteration for a centralised model with one agent and one task\medskip}
\label{tab:exp2}
\scriptsize
\renewcommand{\arraystretch}{1.2}
\begin{tabular}{|c|cc|cccc|}
\hline
\multirow{2}{*}{\shortstack{W.H.\ Size\\ $W\times H$}} & 
\multicolumn{2}{c|}{Model Size} &
\multicolumn{4}{c|}{Time per iter.} \\
& states & trans & CPU & GPU & Prism & Storm \\
\hline
3$\times$3 & 334 & 2.17k & \textbf{1e-4} & 0.03 & 0.005 & 0.038\\
6$\times$6 & 4.2k & 18.8k & \textbf{0.004} & 0.038 & 0.025 & 0.058\\
8$\times$8 & 12.9k & 78.5k  & \textbf{0.017} & 0.041 & 0.081 & 0.114\\
10$\times$10 & 30.9k & 187k & 0.048 & \textbf{0.046} & 0.17 & 0.33\\
\hline
\end{tabular}
\vspace{-5mm}
\end{table}

Experiment 2 compared the performance of our implementation and the multi-objective model checking function in Prism and Storm.
This experiment was conducted on a centralised model regarding one agent and one task, which was essentially a standard MOMDP model acceptable by those two tools.
For our implementation, we restricted the MORAP problem to the \emph{verification-only} setting, achieved by replacing Lines~\ref{ln:find-minimum2}-\ref{ln:update-t-tmp}
in Algorithm \ref{alg:scheduler-synthesis}
with a \emph{break} statement to terminate the algorithm.
(Thus, the \emph{break} statement is executed if and only if the verification returns \emph{false}.)
For Prism and Storm, we specified the same problem as an achievabilty query.
The comparison included model checking time only and excluded the model building time.
%
The results, as shown in Table~\ref{tab:exp2}, indicate that our implementation can still achieve competitive performance compared against the existing tools.
It can be seen that, even without utilising the parallelism of the decentralised model, our implementation is an efficient framework for multi-objective probabilistic model checking.

\section{Related Work}
\label{sec:related}
Multi-objective optimisation considers the domain of planning where objectives may be conflicting, and Pareto-optimal solutions are of interest.
These problems are often the focus of multi-objective model checking \cite{roijers2013survey}.
Efficient synthesis of a set of Pareto optimal schedulers maximising expected total rewards for multi-objective model checking are covered in \cite{chatterjee2006markov,delgrange2020simple,etessami2007multi,forejt2012pareto,roijers2014bounded,forejt2011quantitative}. 
While step and reward bounded reachability probabilities are covered in \cite{forejt2012pareto,hartmanns2020multi}.
Recently, a computationally efficient procedure for multi-objective model checking of long-run average and total mixed rewards is presented in \cite{quatmann2021multi}, a generalisation of \cite{forejt2012pareto}. 
Our point-oriented Pareto computation is a new method complementing the existing multi-objective queries in \cite{forejt2012pareto} specifically targeting scalability in multi-agent systems. Different from existing approaches, if a given threshold point is non-feasible, our algorithm computes a Pareto-optimal point which is nearest the given point.

GPU acceleration for MOMDP is studied in \cite{chowdhury2022gpu}, but is problem specific without task verification. 
A parallel GPU accelerated sparse value iteration algorithms are presented in \cite{sapio2018efficient,wu2016gpu}. The implementation in \cite{sapio2018efficient} is similar to ours, particularly value iteration within Line \ref{ln:find-scheduler-1} of Algorithm \ref{alg:find-supporting-hyperplane} including the reduce kernel operation for action comparison,
but does not consider multiple objectives, or task specification. 
The GPU acceleration considered in \cite{wu2016gpu}, requires specific strongly connected component topologies to achieve optimal parallel performance.
In contrast, our parallel implementation takes advantage of multi-agent and task factorisation, and are always present in our problem.

The approach in \cite{faruq_simultaneous_2018} aims to reduce the redundant complexity in the multi-agent MDP \cite{boutilier1996planning} for problems in which agents do not collaborate on tasks, only that an agent optimally completes its allocated tasks. 
We consider the classical random assignment problem for which agents may only work independently on a single task.
The model generated in \cite{faruq_simultaneous_2018} is not suitable for solving our problem as no mechanism exists for tracking which agents have been assigned a particular task. 
Moreover, by decentralising the task allocation model, this work achieves linear scalability with respect to the numbers of agents and tasks.

\section{Conclusion}
\label{sec:conclusion}

In this paper, we presented an approach addressing the problem of simultaneous random task assignment and planning in an MAS under multi-objective constraints.
We demonstrated that our problem is convex and solvable in polynomial time, and that an optimal random assignment and schedulers can be computed in a decentralised way. 
We provided a hybrid GPU-CPU multi-objective model checking framework which optimally manages the computational load on GPU devices and multiple CPU-cores. We conducted two experiments to show that decentralising the problem results in a parallel implementation which can achieve linear scaling and significant run time improvement. Our experiments also demonstrated that the multi-objective model checking performance of our framework is competitive compared with the probabilistic model checkers Prism and Storm. 
Future work consists of further optimisation of the implementation utilising CUDA streams to alleviate the PCI bottleneck for small individual agent-task models. 
We are also interested to extend our MORAP problem to include tasks expressed as $\omega$-regular temporal properties and limiting behaviours
(e.g., mean pay-offs). 

\bibliographystyle{splncs04}
\bibliography{bib1,bib3}

\newpage
\appendix

\section{Supplementary Materials and Proofs}

\subsection{Formal Definition of Co-Safe LTL}
The semantic relationship, denoted $\sigma\models \varphi$ for any $\sigma\in \Sigma^\omega $ and any LTL formula $\varphi$, is standard. 
For any $\varsigma \in \Sigma^*$ and $\sigma\in \Sigma^\omega $, let $\varsigma\cdot\sigma$ denote the concatenation of $\varsigma$ and $\sigma$.
$\varphi$ is \emph{safe} if the following holds: there is a subset of $\Sigma^*$, denoted $pref_{bad}(\varphi)$, such that if $\varsigma \in pref_{bad}(\varphi)$ then $\varsigma\cdot\sigma \in \Sigma^\omega \backslash \models \varphi$ for all $\sigma\in \Sigma^\omega$;
$\varphi$ is \emph{co-safe} if the following holds: there is a subset of $\Sigma^*$, denoted $pref_{good}(\varphi)$, such that if $\varsigma \in pref_{good}(\varphi)$ then $\varsigma\cdot\sigma \models \varphi$ for all $\sigma \in \Sigma^\omega$.
\begin{lemma}
[\cite{KupfermanOrna2001MCoS}]
\label{prop:co-safe_DFA_LTL_eq}
For any co-safe LTL formula $\varphi$, there is a DFA $\dfa $ whose accepting locations are sinks such that $pref_{good}(\varphi)= acc(\dfa )$.
\end{lemma} 

\subsection{Supplementary Materials for Geometric and Stochastic Matrix}
Recall that a set $C\subseteq \mathbb{R}^m$ for some $m$ is a \emph{convex polytope} if it is a set of all convex combinations of some finite set of vectors. 
A \emph{face} of a convex polytope $C$ is a subset $H\subseteq C$ such that there is a vector $\vect{v}$ such that $\vect{u}\cdot \vect{v}\geq \vect{u}'\cdot \vect{v}$ for all $\vect{u}\in H,\vect{u'}\in C$.
We recall the following property for convex polytopes (which is also an alternative definition of faces).
\begin{lemma}[\cite{forejt2012pareto}] \label{lem:polytope}
Let $C$ be a convex polytope. For any vector $\vect{v}$, there is a face $H$ such that $\vect{u}\cdot \vect{v}=\vect{u}'\cdot \vect{v}$ and $\vect{u}\cdot\vect{v}> \vect{u}''\cdot \vect{v}$ for all $\vect{u},\vect{u'}\in H$ and $\vect{u}''\in C\backslash H$.   
\end{lemma}

We also recall the following well-known property for bistochastic matrices.
\begin{lemma}\label{lem:bistochastic}
Given $\vect{U}\in \mathbb{R}^{n\times n}$, the problem of maximising $\sum_{1\leq i,j\leq n} u_{i,j} x_{i,j}$ such that $\vect{X}$ is bistochastic has an optimal solution $\vect{X}^*$ which is a permutation matrix.
\end{lemma}

\subsection{Proofs for the Centralised MDP}
We present the complete proof of Theorem~\ref{thm:main} in Section~\ref{sec:convexity}.
\begin{proof}[Theorem~\ref{thm:main}]
Assume that there are a bistochastic matrix $(x_{i,j})_{i\in I, j\in J}$ and $Sch(\mathcal{M}_{i\otimes j})$ for each $i\in I,j\in J$ such that  
the inequalities in Def.~\ref{def:MORAP} are satisfied.
Let $y_{i,j,\sharp} = x_{i,j}/(x_{i,j}+\sum_{i'>i,i'\notin \sharp}x_{i',j})$.
Let $\schr^{\tlab }\in {Sch}(\tmdp )$ such that for all pairs $(i,j)\in I\times J$ and $\sharp\subseteq I$: 
\begin{itemize}
    \item $\schr^{\tlab }(s_{i},q_{j},\sharp)(a) =\schr_{i,j}(s_{i},q_{j})(a)$ for all $a\in A_i(s_{i})$ if $i\in \sharp$ and $q_j\notin Q_{j,F}\cup Q_{j,R}$
    \item $\schr^{\tlab }(s_{i},q_{j},\sharp) (b_3) = 1 $ if $i\in \sharp$ and $j\in Q_{j,F}\cup Q_{j,R}$
    \item $\schr^{\tlab }(s_{i,0},q_{j,0},\sharp) (b_1)= y_{i,j,\sharp} $ and $\schr^{\tlab }(s_{i,0},q_{j,0},\sharp) (b_2)= 1- y_{i,j,\sharp} $ if $i\notin \sharp$
\end{itemize}
$\schr^{\tlab }$ is a well-defined scheduler and satisfies the condition in Theorem~\ref{thm:main}.

Conversely, assume that there is $\schr^{\tlab }\in Sch(\tmdp )$ such that the condition in Theorem~\ref{thm:main} holds.
Let $I^{-i} = \{i'\in I \mid i'\neq i\}$.
For each $(i,j)\in I\times J,\sharp\subseteq I^{-i}$, let $x_{i,j,\sharp} = \schr^{\tlab }(s_{i,0},q_{j,0},\sharp)(b_1)z_{i,j,\sharp}$ where $z_{i,j,\sharp}$ is the probability of reaching the tuple in $(s_{i,0},q_{j,0}, \sharp)$ in $\tmdp$ under $\schr^\tlab$, and let $x_{i,j} = \sum_{\sharp'\subseteq I^{-i}} x_{i,j,\sharp'}$.
The scheduler $\schr_{i,j} \in {Sch}(\mdp_{i\otimes j})$ is defined as follows: $\schr_{i,j}(s,q)(a) = \sum_{\sharp\subseteq I^{-i}} \schr^\tlab(s,q,\sharp)(a)x_{i,j,\sharp}$ for all $(s,q)\in S_i\times Q_j$ and $a\in A_i(s)$.
The matrix $(x_{i,j})_{i\in I,j\in J}$ is bistochastic.
To see this, an (informal) argument for this is as follows: Let $\ipath$ be an arbitrary of $\tmdp$ which starts from its initial state $(s_{1,0},q_{1,0},\emptyset)$.  For each $i\in I$, $\ipath$ traverses $(s_{i,0},q_{j,0}, \sharp)b_1(s_{i,0},q_{j,0}, \sharp\cup\{i\})$ \emph{exactly once} for some $j\in J$.  For each $j\in J$, $\ipath$ traverses $(s_{i,0},q_{j,0}, \sharp)b_1(s_{i,0},q_{j,0}, \sharp\cup\{i\})$ for some $i$ \emph{exactly once} for some $i\in I$.
Thus, the MORAP problem is feasible.
\qed\end{proof}

We present two important properties for $\tmdp$ which are used in the subsequent proofs.
Let $|I|=|J|=n$.
Recall that $\mathscr{C}_0\doteq \{ (\mathbf{E}^{\tmdp[\rho_k],\schr})_{1\leq k\leq 2n}\mid \schr\in {Sch}(\tmdp)\}$ and $\mathscr{C}$ is the downward closure of $\mathscr{C}_0$.
The reward-finiteness assumption implies that $\mathscr{C}_0$ is non-empty and bounded.
The first property is a fundamental property of multi-objective MDPs.
\begin{lemma}[\cite{etessami2007multi}]
$\mathscr{C}_0$ is a convex polytope with a finite number of faces.
\end{lemma}
In the worst case, the number of faces in $\mathscr{C}_0$ is exponential in the size of $\tmdp$ and the number of objectives ($2n$ here).

Let $\rho$ (resp., $\rho_{i,j}$) be a reward structure for $\tmdp$ (resp., $\mdp_{i\otimes j}$). We write $\rho\sim \rho_{i,j}$ if $\rho=\rho_{i,j}$ when restricting $\rho$ to the domain of $\rho_{i,j}$.

\begin{lemma}\label{lem:ct-mdp-prop}
Given any $\rho$ for $\tmdp$ and $\schr\in {Sch}(\tmdp)$, there is a bistochastic matrix $(x_{i,j})_{i\in I, j\in J}$ and $\schr_{i,j} \in {Sch}(\mdp_{i\otimes j})$ such that 
\begin{equation}\label{eq:ct-mdp-prop}
    \mathbf{E}^{\tmdp [\rho],\schr}=\sum_{1\leq i,j\leq n} x_{i,j} \mathbf{E}^{\mdp_{i\otimes j}[\rho_{i,j}],\schr_{i,j}}
\end{equation}
where $\rho\sim \rho_{i,j}$ for all $1\leq i,j\leq n$. 
Moreover, there is a permutation matrix $(x_{i,j})_{i\in I, j\in J}$ such that
\begin{equation}\label{eq:ct-mdp-prop-opt}
    \max_{\schr\in {Sch}(\tmdp)}\mathbf{E}^{\tmdp [\rho],\schr}=\sum_{1\leq i,j\leq n} x_{i,j} \mathbf{E}^{\mdp_{i\otimes j}[\rho_{i,j}],\schr_{i,j}^*}
\end{equation}
where $\schr_{i,j}= \argmax_{\schr'}\mathbf{E}^{\mdp_{i\otimes j}[\rho_{i,j}],\schr'}$ and $\rho\sim \rho_{i,j}$ for all $1\leq i,j\leq n$.
\end{lemma}
\begin{proof}
First, we follow the second part of the proof of Theorem~\ref{thm:synthesis} to construct a bistochastic matrix  $(x_{i,j})_{i\in I, j\in J}$ and $\schr_{i,j} \in {Sch}(\mdp_{i\otimes j})$ for all $1\leq i,j\leq n$. Eq.~\eqref{eq:ct-mdp-prop} can be derived by the standard probabilistic model checking method for DTMCs and expected total rewards \cite{baier_principles_2008}. 
Moreover, as each $x_{i,j}$ is non-negative, to maximise $\mathbf{E}^{\tmdp [\rho],\schr}$, we need to maximise $\mathbf{E}^{\mdp_{i\otimes j}[\rho_{i,j}],\schr_{i,j}}$ for all $i,j$.
We can fix $\schr_{i,j}^*= \argmax_{\schr'}\mathbf{E}^{\mdp_{i\otimes j}[\rho_{i,j}],\schr'}$ for all $i,j$, by Lemma~\ref{lem:bistochastic} there is a permutation matrix $(x_{i,j})_{i\in I, j\in J}$ maximising $\sum_{1\leq i,j\leq n} x_{i,j} \mathbf{E}^{\mdp_{i\otimes j}[\rho_{i,j}],\schr_{i,j}^*}$.
\qed\end{proof}

\subsection{Proofs for Algorithm~\ref{alg:scheduler-synthesis} and  Algorithm~\ref{alg:find-supporting-hyperplane}}

\begin{lemma}\label{lem:weight-vector}
$\vect{w}$ computed in Line~\ref{ln:find-w-1} in Alg.~\ref{alg:scheduler-synthesis} is a weight vector, i.e., $\vect{w}\geq 0$ and $\|\vect{w}\|_1=1$.
\end{lemma}
\begin{proof}
Recall that $\|\vect{x}\|^2= \langle\vect{x},\vect{x}\rangle=\vect{x}^T\vect{M}\vect{x}$.
Let $\vect{w}_{\bot}=\vect{M}(\vect{t}-\vect{t}_{\uparrow})$. As $\vect{t}\neq \vect{t}_{\uparrow}$, $\vect{w}_{\bot}\neq \vect{0}$ (otherwise $(\vect{t}-\vect{t}_{\uparrow})^T\vect{M}(\vect{t}-\vect{t}_{\uparrow})=0$ which violating the positive definiteness of $\vect{M}$).
As $\|\vect{t}-\vect{t}_{\uparrow}\|$ is the distance between $\vect{t}$ and ${down}(\Phi)$, the set $\{\vect{x}\mid \vect{w}_{\bot}\cdot (\vect{x}-\vect{t}_{\uparrow})=0\}$ is a separating hyperplane between $\{\vect{t}\}$ and ${down}(\Phi)$, that is, $\vect{w}_{\bot}\cdot (\vect{y}-\vect{t}_{\uparrow})\leq 0 $ for all $\vect{y}\in {down}(\Phi)$.
As ${down}(\Phi)$ is unbounded from below, $\vect{w}_{\bot}\geq 0$ and there is at least one element $w_i$ in $\vect{w}_{\bot}$ such that $w_i>0$. Thus, $\vect{w}= \vect{w}_{\bot}/\|\vect{w}_{\bot}\|_1$ is a weight vector.
\end{proof}

\begin{lemma}\label{lem:face}
Algorithm~\ref{alg:find-supporting-hyperplane} is correct, namely, $\vect{w}\cdot \vect{x} = \vect{w}\cdot \vect{r}$ defines a supporting hyperplane of $\mathscr{C}$ where $\vect{w}$ is computed in Line~\ref{ln:find-w-1} in Alg.~\ref{alg:scheduler-synthesis} and $\vect{r}$ is returned from Alg.~\ref{alg:find-supporting-hyperplane}.
\end{lemma}
\begin{proof}
Let $\vect{u}$ be any vector in $\mathscr{C}_0$ and $\vect{u} = (\mathbf{E}^{\tmdp[\rho_k], \schr_0} )_{1\leq k\leq 2n}$ for some $\schr_0\in {Sch}(\tmdp)$.
By Lemma~\ref{lem:weight-vector}, $\vect{w}$ is a weight vector.
Let $\schr_{i,j}^* = \argmax_{\schr_{i,j}}\mathbf{E}^{\tmdp [\vect{w}\cdot \vect{\rho}],\schr_{i,j}}$ for all $1\leq i,j\leq n$, where $\vect{\rho}=\{\rho_i\}_{1\leq i\leq 2n}$.
Then,
\begin{equation*}
\begin{aligned}
   & \vect{w}\cdot\vect{u}  \\
  =\ & \mathbf{E}^{\tmdp [\vect{w}\cdot \vect{\rho}],\schr_0}\\
   \leq\ & \max_{\schr}\mathbf{E}^{\tmdp [\vect{w}\cdot \vect{\rho}],\schr}\\
   =\ & \textstyle\sum_{i,j} x_{i,j}\mathbf{E}^{\mdp_{i\otimes j} [\vect{w}\cdot \vect{\rho}],\schr_{i,j}^*} & \text{(for some permutation\ matrix $(x_{i,j})_{1\leq i,j\leq n}$;} \\
   & &  \text{c.f.\ Eq.~\eqref{eq:ct-mdp-prop-opt} in Lemma~\ref{lem:ct-mdp-prop})} \\
   \leq\ & \textstyle\sum_{i,j} \mathrm{I}_{i=f(j)}\mathbf{E}^{\mdp_{i\otimes j} [\vect{w}\cdot \vect{\rho}],\schr_{i,j}^*} & \text{(according to def.\ of $f$)}\\
   =\ & \vect{w} \cdot \vect{r}
\end{aligned} 
\end{equation*}
The last equality also confirms $\vect{r}\in \mathscr{C}_0\subset  \mathscr{C}$ (i.e., $\vect{r}$ is a feasible threshold).
\qed\end{proof}

We now present the complete proof for Theorem~\ref{thm:synthesis}.

\begin{proof}[Theorem~\ref{thm:synthesis}]
We first show the \emph{termination} of Algorithm~\ref{alg:scheduler-synthesis}. 
Assume that the $\ell^{\mathrm{th}}$ iteration of Algorithm~\ref{alg:scheduler-synthesis} is completed for any $\ell>1$. 
By Lemma~\ref{lem:weight-vector}, $\vect{w}\geq 0 $ is a weight vector.
Informally, the algorithm finds a sequence of values for $\vect{t}_{\uparrow}$ (resp.,  $\vect{t}_{\downarrow}$) which move towards (resp., away from) $\vect{t}$ and terminates eventually with $\|\vect{t}_{\uparrow}-\vect{t}_{\downarrow}\|\leq \varepsilon$. The formal proof relies on the following two claims.
\begin{claim}
    If $\vect{w}\cdot\vect{r}> \vect{w}\cdot\vect{t}_{\uparrow}$, then $\vect{r}$ is on a new face of $\mathscr{C}_0$.
\end{claim}
Actually, $\vect{w}\cdot\vect{r}> \vect{w}\cdot\vect{t}_{\uparrow}$ implies that $\vect{w}\cdot\vect{r}> \vect{w}\cdot\vect{u}$ for all $\vect{u}\in {down}(\Phi\backslash\{\vect{r}\})$ (since $\{\vect{u}\in \mathbb{R}^{2n} \mid \vect{w}\cdot\vect{u}=\vect{w}\cdot\vect{t}_{\uparrow}\}$ is a supporting hyperplane for ${down}(\Phi\backslash\{\vect{r}\})$). Thus, under this condition, by Lemma~\ref{lem:face} and Lemma~\ref{lem:polytope}, there is a face $H$ of $\mathscr{C}_0$ such that $\vect{r}\in H$ and $\vect{r}'\notin H$ for all $\vect{r'}\in \Phi\backslash\{\vect{r}\}$; in other words, $\vect{r}$ is on a new face of $\mathscr{C}_0$.

\begin{claim}
    If $\vect{w}\cdot\vect{r}\leq  \vect{w}\cdot\vect{t}_{\uparrow}$, then $\vect{t}_{\downarrow} =\vect{t}_{\uparrow}$.
\end{claim}
Actually, as $\vect{t}_{\uparrow}\in{down}(\Phi)$, $\vect{w}'\cdot\vect{r}'\geq \vect{w}'\cdot\vect{t}_{\uparrow}$ for all $(\vect{w}',\vect{r}')\in\Lambda$ by Lemma~\ref{lem:face}. 
Thus, the condition $\vect{w}\cdot\vect{r}\leq  \vect{w}\cdot\vect{t}_{\uparrow}$ equals to $\vect{w}\cdot\vect{r}= \vect{w}\cdot\vect{t}_{\uparrow}$. 
Under this condition, $\vect{t}_{\uparrow}$ and $\vect{t}_{\downarrow}$ are the unique vectors such that
\begin{equation*}
\begin{aligned}
   \vect{t}_{\uparrow} 
   =\ & \argmin_{\vect{z}\in \mathbb{R}^{2n},\vect{w}\cdot\vect{t}_{\uparrow}= \vect{w}\cdot \vect{z} } \|\vect{t}-\vect{z}\|&  \\
   =\ & \argmin_{\vect{z}\in \mathbb{R}^{2n},\vect{w}\cdot\vect{r}\geq \vect{w}\cdot \vect{z} } \|\vect{t}-\vect{z}\| \\
   =\ & \argmin_{\vect{z}\in \mathbb{R}^{2n},\vect{w}'\cdot\vect{r}'\geq \vect{w}'\cdot \vect{z},\forall (\vect{w}',\vect{r}')\in\Lambda } \|\vect{t}-\vect{z}\| \\
   =\ & \vect{t}_{\downarrow}
\end{aligned}
\end{equation*}
The second equality holds because $\{\vect{x}\mid\vect{w}\cdot\vect{r}=\vect{w}\cdot\vect{x}\}$ is a separating hyperplane between $\vect{t}$ and $\{\vect{z}\mid\vect{w}\cdot\vect{r}\geq \vect{w}\cdot\vect{z}\}$.

A direct consequence of the first claim above is that the inequality $\vect{w}\cdot\vect{r}> \vect{w}\cdot\vect{t}_{\uparrow}$ cannot hold for infinitely many iterations. Thus, either the algorithm terminates or $\vect{w}\cdot\vect{r}\leq  \vect{w}\cdot\vect{t}_{\uparrow}$ in some iteration. 
If $\vect{w}\cdot\vect{r}\leq  \vect{w}\cdot\vect{t}_{\uparrow}$, the second claim above guarantees the termination of Algorithm~\ref{alg:scheduler-synthesis} for any $\varepsilon\geq 0$.

\emph{Property (i)} is obvious as $\vect{t}_{\uparrow}\in {down}(\Phi)\subseteq  \mathscr{C} $.
%
\emph{Property (ii)} holds by observing that if $\vect{t}\in \mathscr{C}$ then the condition in Line~\ref{ln:condition1} of Algorithm~\ref{alg:scheduler-synthesis} is always false.
%
For \emph{property (iii)}, the inequality $\min_{\vect{u}\in\mathscr{C}}\|\vect{t}- \vect{u}\| \leq \min_{\vect{u}\in{down}(\Phi)}\|\vect{t}- \vect{u}\| = \|\vect{t}-\vect{t}_{\uparrow}\|$ holds after the first iteration of Algorithm~\ref{alg:scheduler-synthesis}.
On the other hand, $\|\vect{t}-\vect{t}_{\downarrow}\|\leq \|\vect{t}- \vect{u}\|$ for all $\vect{u}\in\mathscr{C}$. 
The inequality holds because any $\vect{u}\in\mathscr{C}$ satisfies the constraints in Line~\ref{ln:find-minimum2}.
Thus $\|\vect{t}-\vect{t}_{\downarrow}\|= \min_{\vect{u}\in\mathscr{C}}\|\vect{t}- \vect{u}\|$.
\qed\end{proof}

\begin{proof}[Corollary~\ref{cor:}]
 Property (i), i.e.\ $\vect{t}_{\uparrow}=\vect{t}_{\downarrow}$, follows immediately from the termination condition. One direction of property (ii) is just property (ii) in Theorem~\ref{thm:synthesis}. For the other direction, suppose  $\vect{t}_{\downarrow}=\vect{t}$. Then, $\vect{t}=\vect{t}_{\uparrow}\in {down}(\Phi)\subseteq \mathscr{C}$.
 For property (iii), if $\vect{t}\notin\mathscr{C}$ then $\vect{t}_{\uparrow}=\vect{t}_{\downarrow}$ and property (iii) in Theorem~\ref{thm:synthesis} implies $\vect{t}_{\downarrow}$ is on the Pareto curve of $\mathscr{C}$.
\end{proof}

\subsection{Policy and Value Iterations in Algorithm~\ref{alg:find-supporting-hyperplane}}
The method for computing Line~\ref{ln:find-scheduler-1}, and that for Lines~\ref{ln:verification-1} and \ref{ln:verification-2} in Alg.\,\ref{alg:find-supporting-hyperplane} are included in Algorithm~\ref{alg:policy-iteration} and \ref{alg:value-iteration}, respectively.
\begin{algorithm}[t]
\KwIn{$\mathcal{M}_{i\otimes j}$, $\vect{w}$, $\vect{\rho}$, $\varepsilon_1>0$\\}
\KwOut{$\schr=\argmax_{\schr}  \mathbf{E}^{\mathcal{M}_{i\otimes j}[\vect{w}\cdot \vect{\rho}],\schr}$, $\mathbf{E}^{\mathcal{M}_{i\otimes j}[\vect{w}\cdot \vect{\rho}],\schr}$\\}
Initialise $\schr$\;
$\vect{x}:=\vect{0}$; $\vect{y}:=\vect{0}$\;
$\mathrm{policy\text{-}stable} := \mathrm{true}$\;
\While{not $\mathrm{policy\text{-}stable}$}{
    \ForEach{$(s,q)\in {S}_{i,j}$}{
        ${y}_{s,q}:= \max_{a\in A_i(s)}
        [(\vect{w}\cdot \vect{\rho})(s,q,a) + \sum_{(s',q')\in {S}_{i,j}} {P}_{i,j}(s,q,a,s',q')\cdot x_{s',q'}]$\label{ln:maxa}\;
        \If{$|y_{s,q}-x_{s,q}|>\varepsilon_1$}{
        $\mathrm{policy\text{-}stable} := \mathrm{false}$\label{ln:policy_update_condition}\;
        $\schr(s,q):=a$ where $a$ is from
        Line~\ref{ln:maxa}\;}
        $x_{s,q}:=y_{s,q}$\;
    }
  }
\label{ln:end_first_block}
\KwRet{$\schr$, $y_{s_{j,0},q_{j,0}}$}
\caption{Computing Line~\ref{ln:find-scheduler-1} in Alg.\,\ref{alg:find-supporting-hyperplane} by policy-iteration \label{alg:policy-iteration}}
\end{algorithm}

\begin{algorithm}[t]
\KwIn{$\mathcal{M}_{i\otimes j}$, $\schr_{i,j}$, $\rho$, $\varepsilon_2>0$\\}
\KwOut{$\mathbf{E}^{\mathcal{M}_{i\otimes j}[\rho],\schr_{i,j}}$\\}
$\vect{x}:=\vect{0}$; $\vect{y}:=\vect{0}$\;
    $\mathrm{value\text{-}stable} :=\mathrm{true}$\;
    \While{not $\mathrm{value\text{-}stable}$}{
      \ForEach{$(s,q)\in {S}_{i,j}$}{
        $y_{s,q}:= {\rho}(s,q,\schr_{i,j}(s,q)) + \sum_{(s',q')\in {S}_{i,j}} {P}_{i,j}(s,q,\schr_{i,j}(s,q),s',q')\cdot x_{s',q'}$\;
        \If{$|y_{s,q} - x_{s,q}|>\varepsilon_2$}{$\mathrm{value\text{-}stable} :=\mathrm{false}$\;}
        $x_{s,q} := y_{s,q}$\;
      }
}
\KwRet{$y_{s_{i,0},q_{j,0}}$}
\caption{Computing Line~\ref{ln:verification-1} and Line~\ref{ln:verification-2} in Alg.\,\ref{alg:find-supporting-hyperplane} by value-iteration\label{alg:value-iteration}}
\end{algorithm}

\end{document}